\documentclass{llncs}

\newcommand{\nop}[1]{}

\title{(Total) Vector Domination for Graphs \\ with Bounded Branchwidth}
\author{Toshimasa Ishii~\inst{1} \and Hirotaka Ono~\inst{2}  \and 
Yushi Uno~\inst{3}}
\institute{Graduate School of Economics and Business Administration, 
Hokkaido University, Sapporo 060-0809, Japan
\and 
Department of Economic Engineering, Faculty of Economics, 
Kyushu University, Fukuoka 812-8581, Japan 
\and 
Department of Mathematics and Information Sciences, Graduate School of
Science, Osaka Prefecture University, Sakai 599-8531, Japan} 
\begin{document}
\maketitle

\begin{abstract}
Given a graph $G=(V,E)$ of order $n$ and an $n$-dimensional non-negative  
vector $d=(d(1),d(2),\ldots,d(n))$, called demand vector, the vector 
domination (resp., total vector domination) is the problem of finding a 
minimum $S\subseteq V$ such that every vertex $v$ in $V\setminus S$ 
(resp., in $V$) has at least $d(v)$ neighbors in $S$. The (total)
 vector domination is a generalization of many dominating set type
 problems, e.g., the dominating set problem, the $k$-tuple dominating
 set problem (this $k$ is different from the solution size),  
 and so on, and its approximability and inapproximability have been
 studied under this general framework. In this paper, we show that a
(total) vector domination of graphs with bounded branchwidth can be
solved in polynomial time. This implies that the problem is 
polynomially solvable also for graphs with bounded treewidth. 
Consequently, the  (total) vector domination problem for 
a planar graph is subexponential fixed-parameter tractable with respect
 to $k$, where $k$ is the size of solution.  
\end{abstract}

\section{Introduction}
Given a graph $G=(V,E)$ of order $n$ and an $n$-dimensional non-negative
vector $d=(d(1),d(2),\ldots,d(n))$, called {\em demand vector}, the {\em
vector domination} (resp., {\em total vector domination}) is the problem of finding a minimum
$S\subseteq V$ such that every vertex $v$ in $V\setminus S$ (resp., in
$V$) has at least $d(v)$ neighbors in $S$. These problems were introduced by
\cite{HPV99}, and they contain many existing problems, such as the
minimum dominating set and the $k$-tuple dominating set 
problem  (this $k$ is different from the solution size)~\cite{harary2000double,haynes1998domination}, and so 
on. Indeed, by setting $d=(1,\ldots,1)$, the vector domination becomes
the minimum dominating set forms, and by setting $d=(k,\ldots,k)$, the
total vector dominating set becomes the $k$-tuple dominating set. If in
the definition of total vector domination, we replace open neighborhoods
with closed ones, we get the {\em multiple domination}. In this paper, 
we sometimes refer to these problems just as {\em domination problems}. 
Table 1 of \cite{DAM2013} 
summarizes how related problems are represented in the scheme of
domination problems. Many variants of the basic
concepts of domination and their applications have appeared in
\cite{haynes1998domination,haynes1998fundamentals}.  

Since the vector or multiple domination includes the setting of the 
ordinary dominating set problem, it is obviously NP-hard, and further it
is NP-hard to approximate within $(c\log n)$-factor, where $c$ is a 
positive constant, e.g.,
$0.2267$~\cite{alon2006algorithmic,lund1994hardness}. 
As for the approximability, since the domination problems are special
cases of a set-cover type integer problem, it is known that the
polynomial-time greedy algorithm achieves an $O(\log n)$-approximation
factor~\cite{dobson1982worst}; it is already optimal in terms of order.  
We can see further analyses of the approximability and inapproximability
in~\cite{FCT2011,DAM2013}.  

In this paper, we focus on another aspect of designing algorithms for
domination problems, that is, the polynomial-time solvability of 
the domination problems for graphs of bounded treewidth or branchwidth. 
In \cite{betzler2012bounded}, it is shown that the vector domination
problem is $W[1]$-hard with respect to treewidth. This result and 
Courcelle's meta-theorem about MSOL~\cite{courcelle1990monadic} 
imply that the vector domination is unlikely expressible in MSOL; 
it is not obvious to obtain a polynomial time algorithm.  

In this paper, we present a polynomial-time algorithm for the domination
problems of graphs with bounded branchwidth. The branchwidth is a
measure of the ``global connectivity'' of a graph, and is known to be a
counterpart of treewidth. It is known that 
\[
 \max\{bw(G),2\} \le tw(G)+1 \le \max\{3bw(G)/2,2\}, 
\]
where $bw(G)$ and $tw(G)$ denote the branchwidth and treewidth of
graph $G$, respectively~\cite{robertson1991graph}. Due to the linear
relation of these two 
measures, polynomial-time solvability of a problem 
for graphs with bounded treewidth implies 
polynomial-time solvability of a problem 
for graphs with bounded branchwidth, and vice versa. 
Hence, our results imply that 
the domination problems (i.e., vector domination, total vector
domination and multiple domination) can be solved in polynomial time for
graphs with bounded treewidth; the 
polynomial-time solvability for all the problems (except the dominating
set problem) in Table 1 of \cite{DAM2013} is newly shown.  
Also, they answer the question by \cite{FCT2011,DAM2013} about  
the complexity status of the domination problems of graphs with
bounded treewidth. 

Furthermore, by using the polynomial-time algorithms for graphs of
bounded  treewidth, 
we can show that these problems for a planar graph are subexponential
fixed-parameter tractable with respect to the size of the solution $k$,
that is, there is an algorithm whose running time is 
$2^{O(\sqrt{k}\log k)}n^{O(1)}$. 
To our best knowledge, these are the first 
fixed-parameter algorithms for the total vector domination and multiple domination, 
whereas the vector domination for planar graphs has been shown to be 
FPT~\cite{raman2008parameterized}. For the latter case, our algorithm
greatly improves the running time. 


Note that the polynomial-time solvability of the vector domination 
problem for graphs of bounded treewidth has been independently shown
very recently~\cite{cicalese2013latency2}. They considered a further
generalization of the vector domination problem, and 
gave a polynomial-time algorithm for graphs of bounded
clique-width. Since $cw(G)\le 2^{tw(G)+1}+1$ holds where $cw(G)$ denotes
the clique-width of graph $G$ (\cite{courcelle2000upper}), their 
polynomial-time algorithm implies the polynomial-time solvability of the
vector domination problem for graphs of bounded treewidth and 
bounded branchwidth. 

\subsection{Related Work}\label{related:subsec}
The dominating set problem itself is one of the most fundamental
graph optimization problems, and it has been intensively and extensively
studied from many points of view. In the sense that the vector or
multiple domination contains the setting of not only the ordinary dominating set
problem but also many variants, there are an enormous number of related
studies. Here we pick some representatives up. 

As a research of the domination problems from the viewpoint of the
algorithm design, Cicalese, Milanic and Vaccaro gave detailed analyses
of the approximability and
inapproximability~\cite{FCT2011,DAM2013}. They also provided
some exact polynomial-time algorithms for special classes of 
graphs, such as complete graphs, trees, $P_4$-free graphs, and threshold
graphs. 

For graphs with bounded treewidth (or branchwidth), the ordinary
domination problems can be solved in polynomial time.  
As for the fixed-parameter tractability, it is known that even the 
ordinary dominating set problem is W[2]-complete 
with respect to solution size $k$; it is unlikely to be 
fixed-parameter tractable~\cite{downey1992fixed}. 
In contrast, it can be 
solved in $O(2^{15.13\sqrt{k}}+n^3)$ time for planar graphs, that
is, it is subexponential fixed-parameter tractable~\cite{fomin2006dominating}. 
The subexponent part comes from the inequality $bw(G)\le 12\sqrt{k}+9$,
where $k$ is the size of a dominating set of $G$. Behind the inequality, 
there is a unified property of parameters, called {\em
bidimensionality}~\cite{demaine2008bidimensionality}. Namely, the
subexponential fixed-parameter algorithm of the dominating set for
planar graphs (more precisely, $H$-minor-free
graphs~\cite{demaine2005subexponential}) is based on the bidimensionality.  

A maximization version of the ordinary dominating set is also
considered. {\em Partial Dominating Set} is the problem of maximizing
the number of vertices to be dominated by using a given number $k$ of
vertices. In \cite{amini2011implicit}, it was shown that partial
dominating set problem is FPT with respect to $k$ for $H$-minor-free
graphs. Later, \cite{fomin2011subexponential} gives a subexponential FPT
with respect to $k$ for apex-minor-free graphs, also a super class of 
planar graphs. Although partial dominating set is an example of problems
to which the bidimensionality theory cannot be applied, they develop a
technique to reduce an input graph so that its treewidth becomes 
$O(\sqrt{k})$. 

For the vector domination, 
a polynomial-time algorithm for graphs of bounded treewidth has 
been proposed very recently~\cite{cicalese2013latency2}, as mentioned
before. In \cite{raman2008parameterized}, it is shown that the vector
domination for 
$\rho$-degenerated graphs can be solved  in $k^{O(\rho k^2)}n^{O(1)}$
time, if $d(v) >0$ holds 
for $\forall v \in V$ (positive constraint). Since any planar graph is
$5$-degenerated, the vector domination for planar graphs is
fixed-parameter tractable  with respect to solution size, 
under the positive constraint. Furthermore, the case where $d(v)$ could
be $0$ for some $v$ can be  easily reduced to the positive case by using
the transformation discussed in \cite{betzler2012bounded}, with increasing the degeneracy
only $1$. It follows that the vector domination for planar graphs is 
FPT with respect to solution size $k$. However, for the total vector
domination and multiple domination, neither polynomial time
algorithm for graphs of bounded treewidth nor fixed-parameter algorithm for 
planar graphs has been known. 

Other than these, several generalized versions of the dominating 
set problem are also studied. $(k, r)$-center problem is the one that asks the
existence of set $S$ of $k$ vertices satisfying that for every vertex
$v\in V$ there exists a vertex $u\in S$ such that the distance between
$u$ and $v$ is at most $r$; $(k,1)$-center corresponds to the ordinary
dominating set. The $(k, r)$-center for planar graphs is shown to be
fixed-parameter tractable with respect to $k$ and
$r$~\cite{demaine2005fixed}. For $\sigma,\rho\subseteq
\{0,1,2,\ldots\}$ and a positive integer $k$, $\exists
[\sigma,\rho]$-dominating set is the problem 
that asks the existence of set $S$ of $k$ vertices  satisfying that
$|N(v)\cap S|\in \sigma$ holds for $\forall v\in S$ and $|N(v)\cap S|\in
\rho$  for $\forall v\in V\setminus S$, where $N(v)$ denotes the open
neighborhood of $v$. If $\sigma=\{0,1,\ldots \}$ and $\rho=\{1,2,\ldots
\}$, $\exists [\sigma,\rho]$-dominating set is the ordinary dominating
set problem, and if $\sigma=\{0 \}$ and $\rho=\{0,1,2,\ldots \}$, it is
the independent set. In \cite{chapelle2010parameterized}, the
parameterized complexity of $\exists [\sigma,\rho]$-dominating set with
respect to treewidth is also considered.

\subsection{Our Results}
Our results are summarized as follows: 
\begin{itemize}
 \item We present a polynomial-time algorithm for the vector domination
       of graph $G=(V,E)$ with bounded branchwidth. The running time is
       roughly $O(n^{6bw(G)+2})$. 
 \item We present polynomial-time algorithms for the total vector
       domination and multiple domination of graph $G$ with bounded
       branchwidth. The running time is roughly $O(2^{9bw(G)/2}$
       $n^{6bw(G)+2})$.  
 \item 
       Let $G$ be a planar graph. 
       Then, we can check in
       $O(n^4+\min\{k,d^*\}^{40\sqrt{k}+34}n)$ time
       whether $G$ has a vector dominating set with cardinality at most $k$
       or not, where $d^*=\max\{d(v)\mid v \in V\}$.
 \item 
       Let $G$ be a planar graph. 
       Then, we can check in
       $O(n^4+2^{30\sqrt{k}+51/2}\min\{k,d^*\}^{40\sqrt{k}+34}n)$ time
       whether $G$ has a total vector dominating set and a multiple dominating
       set  with cardinality at most $k$ or not. 
\end{itemize}
It should be noted that it is actually possible to design directly 
polynomial time algorithms for graphs with bounded treewidth, 
but they are slower than the ones for graphs with bounded branchwidth; 
this is the reason why we adopt the branchwidth instead of the
treewidth. 

As far as the authors know, the second and fourth results give the first
polynomial time algorithms and the first fixed-parameter algorithm for
the total vector 
domination and multiple domination of graphs with bounded branchwidth
(or treewidth) and planar graphs, respectively. As for the vector domination,  
we give an $O(n^{6bw(G)+2})$-time algorithm, 
whose running time is $O(n^{6(tw(G)+1)+2})$ in terms of the treewidth, 
whereas the recent paper~\cite{cicalese2013latency2} gives 
an $O(cw(G)|\sigma|(n+1)^{5cw(G)})$-time algorithm, where $|\sigma|$ 
is the encoding length of $k$-expression used in the algorithm, 
and is bounded by a polynomial in the input size for fixed $k$. 
Since $cw(G)\le 2^{tw(G)+1}+1$ holds, this is an
$O(2^{tw(G)+1}|\sigma|(n+1)^{5\cdot 2^{tw(G)+1}})$-time algorithm. 

Also, the third result shows that the vector domination of planar graphs
is subexponential FPT with respect to $k$, and it greatly improves the
running time of existing $k^{O(k^2)}n^{O(1)}$-time algorithm
(\cite{raman2008parameterized}). 
It was shown in \cite{cai2003subexponential} that for the ordinary
dominating set problem 
(equivalently,  the
vector domination (or multiple domination) with $d=(1,1,\ldots,1)$) in
planar graphs,
there is  no $2^{o(\sqrt{k})}n^{O(1)}$-time algorithm 
unless the Exponential Time Hypothesis 
(i.e., 
the assumption that there is no $2^{o(n)}$-time algorithm for
$n$-variable 3SAT \cite{impagliazzo2001ETH})
fails. 
Hence, in this sense, our algorithm in third result (or the fourth
results for the multiple domination) is optimal if $d^*$
is a constant.

The third and fourth results give
subexponential fixed-parameter algorithms of the domination problems 
for planar graphs. It should be noted that the domination problems 
themselves do not have the bidimensionality, mentioned in the previous
subsection, due to the existence of the vertices with demand $0$. 
Instead, by reducing irrelevant vertices, we obtain a similar inequality
about the branchwidth and the solution size of the domination problems,
which leads to the subexponential fixed-parameter algorithms.  

\medskip 

The remainder of the paper is organized as follows. In Section 2, we
introduce some basic notations and then explain the
branch decomposition. Section 3 is the main part of the paper, and
presents our dynamic programming based algorithms for the considered
problems. Section 4 explains how we extend the algorithms of Section 3
to fixed-parameter algorithms for planar graphs.


\section{Preliminaries}

A graph $G$ is an ordered pair of its vertex set $V(G)$ and edge set $E(G)$ 
and is denoted by $G=(V(G),E(G))$. Let $n=|V(G)|$ and $m=|E(G)|$. 
We assume throughout this paper that all graphs are undirected, 
and simple, 
unless otherwise stated. 
Therefore, an edge $e\in E(G)$ is an unordered pair of vertices $u$ and $v$, 
and we often denote it by 
$e=(u,v)$.
Two vertices $u$ and $v$ are {\em adjacent} if 
$(u,v)\in E(G)$.
For a graph $G$, the ({\em open}) {\em neighborhood} of a vertex $v\in V(G)$ 
is the set $N_G(v)=\{u\in V(G)\mid (u,v)\in E(G)\}$, 
and the {\em closed neighborhood} of $v$ 
is the set $N_G[v]=N_G(v)\cup\{v\}$. 


For a graph $G=(V,E)$, let $d=(d(v) \mid v \in V)$ be
an $n$-dimensional non-negative
vector. 
Then, we call a set $S \subseteq V$ of vertices
a {\em $d$-vector dominating set} (resp., {\em $d$-total vector dominating set}) 
if $|N_G(v) \cap S|\geq d(v)$ holds for
 every vertex $v \in V-S$ (resp., $v \in V$).
We call a set $S \subseteq V$ of vertices
a {\em $d$-multiple dominating set} 
if $|N_G[v] \cap S|\geq d(v)$ holds for
 every vertex $v \in V$.
We may drop $d$ in these notations if there are no confusions.

\subsection{Branch decomposition}

A {\em branch decomposition} of a graph $G=(V,E)$  is
 defined as a pair $(T=(V_T,E_T), \tau)$
 such that
(a) $T$ is a tree with $|E|$ leaves
in which every non-leaf node has degree 3, and
(b) $\tau$ is a bijection from $E$ to the set of leaves of $T$.
Throughout the paper, we shall use {\em node} to denote an element in $V_T$
for distinguishing it from an element in $V$.

For an edge $f$  in $T$, 
let $T_f$ and $T-T_f$ be two trees obtained from $T$ by removing $f$, and
 $E_f$ and $E-E_f$ be  two sets of edges in $E$ such that $e \in E_f$
if and only if
$\tau(e)$ is included in $T_f$.
The {\em order function} $w: E(T) \to 2^{V}$ is defined as follows:
for an edge $f$ in $T$, a vertex $v \in V$ belongs to $w(f)$ if and only if 
there exist an edge in $E_f$ and an edge in $E-E_f$
which  are both incident to $v$.
The {\em width} of a branch decomposition $(T, \tau)$ is
$\max\{|w(f)| \mid f \in E_T\}$,
and the {\em branchwidth} of $G$, denoted by $bw(G)$,
is the minimum width over all
branch decompositions of $G$.

In general, computing the branchwidth of a given graph 
is NP-hard \cite{seymour1994call}.
On the other hand, 
Bodlaender and Thilikos~\cite{bodlaender1997constructive}
 gave a linear time algorithm which
checks 
 in linear time 
whether the branchwidth of a given
graph
is at most $k$ or not, and if so, 
outputs a branch decomposition of minimum width, for any fixed $k$.
Also, as shown in the following lemma, it is known that for planar graphs, 
it can be done in polynomial time for any given $k$,
where a graph is called {\em planar}
 if it can be drawn
in the plane without generating a crossing by two edges.

\begin{lemma}\label{branch:lem}{\rm (}\cite{seymour1994call}{\rm )}
Let $G$ be a planar graph.
Then, it can be checked in $O(n^2)$ time  whether 
$bw(G)\leq k$ or not for a given integer $k$. 
Also, we can construct a branch decomposition of $G$ 
with width $bw(G)$ in $O(n^4)$ time. \qed
\end{lemma}

Here, we introduce  the following basic properties about branch decompositions,
which will be utilized in the subsequent sections.

\begin{lemma}\label{branch1:lem}
Let $(T,\tau)$ be a branch decomposition  of $G$.

\noindent
$(i)$ For a tree $T$, let $x$ be a non-leaf node and
 $f_i=(x,x_i)$, $i=1,2,3,$ be an edge incident to $x$
$($note that the degree of $x$ is three$)$.
Then,  $w(f_i)-w(f_j)-w(f_k)=\emptyset$ 
for every $\{i,j,k\}=\{1,2,3\}$.
Hence, $w(f_i) \subseteq w(f_j) \cup w(f_k)$.


\noindent 
$(ii)$ Let $f$ be an edge of $T$, 
$V_1$ be the set of all end-vertices of edges in $E_f$,
and
 $V_2$ be the set of all end-vertices of edges in $E-E_f$.
Then, $(V_1-w(f))\cap (V_2-w(f))=\emptyset$ holds.
Also, 
there is no edge in $E$ connecting a vertex in $V_1-w(f)$
and a vertex in  $V_2-w(f)$.
\end{lemma}
\begin{proof}
(i) Without loss of generality, 
assume that $E_{f_1}\cap E_{f_2}=\emptyset$,
$E_{f_2}\cap E_{f_3}=\emptyset$, 
$E_{f_3}\cap 
E_{f_1}=\emptyset$, and
 $E_{f_1} \cup E_{f_2} \cup E_{f_3}=E$.
 Let $v \in w(f_1)$ be a vertex. 
From the definition of $w(f_1)$, there exist two edges
$e \in E_{f_1}$ and $e' \in E-E_{f_1}$ such that
both of $e$ and $e'$ are incident to $v$.
If $e' \in E_{f_2}$ (resp., $e' \in E_{f_3}$), then $v \in w(f_2)$ 
(resp., $v \in w(f_3)$) also holds.
Thus, we can observe that there is no vertex in 
 $w(f_i)-w(f_j)-w(f_k)$ 
for every $\{i,j,k\}=\{1,2,3\}$.


(ii) Assume by contradiction that there exists a vertex 
$v \in (V_1-w(f)) \cap (V_2-w(f))$.
From definition of $V_1$ and $V_2$, then 
there exists an edge $e_1 \in E_f$ and an edge $e_2 \in E-E_f$
such that both of $e_1$ and $e_2$ are incident to $v$.
From the existence of $e_1$ and $e_2$ and the definition of $w(f)$,
it follows that  $w(f)$ also contains $v$,
which contradicts $v \notin w(f)$.

Assume by contradiction that there exists an edge $e=(u_1,u_2) \in E$ such that
$u_1 \in V_1-w(f)$ and   $u_2 \in V_2-w(f)$.
If we assume that $e \in E_f$ without loss of generality,
then   $u_2 \in V_1-w(f)$ also holds,
which contradicts $ (V_1-w(f)) \cap (V_2-w(f))=\emptyset$.
\qed
\end{proof}

\section{Domination problems in graphs of bounded branchwidth}

In this section, we propose  dynamic programming algorithms
for the vector domination problem,
the total vector domination problem, and the multiple domination problem, 
by utilizing a branch decomposition of a given graph.
The techniques  are based on the one developed by
Fomin and Thilikos
for solving the dominating set 
problem with bounded branchwidth \cite{fomin2006dominating}. 
Throughout this section, for a given graph $G=(V,E)$,
the demand of each vertex $v \in V$ is denoted by $d(v)$,
and 
let $d^*=\max\{d(v)\mid v \in V\}$.

\subsection{Vector domination}\label{algo1:subsec}


In this subsection, we consider the vector domination problem,
and
show the following theorem.

\begin{theorem}\label{algo1:th}
If  a branch decomposition of $G$ with width $b$ is given, 
a minimum vector dominating set in $G$  can be found in 
  $O((d^*+2)^{b}\{(d^*+1)^2+1\}^{b/2}m)$ time. 
\qed
\end{theorem}

\noindent  Due to the assumption of the above theorem, 
 we need to consider how we obtain a branch decomposition of $G$
 for the completeness of an algorithm of the vector domination problem. 
 For a branch decomposition, there exists an $O(2^{b\lg{27}}n^2)$-time
 algorithm that given a graph $G$, reports $bw(G)\geq b$, or outputs a
 branch decomposition of $G$ with width at most $3b$
 \cite{robertson1995disjointpath,demaine2005subexponential}.  
 Thus, the time to find a branch decomposition with width at most 
 $3bw(G)$ is $O(\log bw(G) 2^{bw(G)\lg{27}}n^2)$ (smaller than the time
 complexity below), and we have the following corollary. 
 \begin{corollary}\label{algo1:cor}
 A minimum vector dominating set in $G$ can be found in 
  $O((d^*+2)^{3bw(G)}\{(d^*$ $+1)^2+1\}^{3bw(G)/2}n^2)$ time. \qed
\end{corollary}

Below, for proving this theorem, we will give a dynamic programming
algorithm for finding a minimum vector dominating set in $G$
in 
  $O((d^*+2)^{b}\{(d^*+1)^2+1\}^{b/2}m)$ time, based
on a branch decomposition of $G$.

Let $(T', \tau)$ be a branch decomposition of $G=(V,E)$ with $b$,
and $w': E(T') \to 2^{V}$ be the corresponding order function.
Let $T$ be the tree from $T'$ by inserting two nodes $r_1$ and $r_2$,
deleting  one arbitrarily chosen edge $(x_1,x_2)\in E(T')$,
adding three new edges $(r_1,r_2)$, $(x_1,r_2)$, and $(x_2,r_2)$;
namely, $T=(V(T')\cup \{r_1,r_2\},E(T')\cup\{(r_1,r_2), (x_1,r_2), (x_2,r_2)\}-\{(x_1,x_2)\})$.
Here, we regard $T$
with a rooted tree by choosing $r_1$ as a root.
Let $w(f)=w'(f)$ for
 every $f \in E(T)\cap E(T')$,
$w(x_1,r_2)=w(x_2,r_2)=w'(x_1,x_2)$,
and
$w(r_1,r_2)=\emptyset$.

Let $f=(y_1,y_2) \in E$ be an edge in $T$ such that $y_1$ is the parent
of $y_2$.
Let  $T(y_2)$ be the subtree of $T$ rooted at $y_2$, 
  $E_f=\{e \in E \mid \tau(e) \in V(T(y_2)) \}$,
and  $G_f$ be the subgraph of $G$ induced by $E_f$. 
Note that $w(f) \subseteq V(G_f)$ holds, since
each vertex in $w(f)$ is an end-vertex of some edge in $E_f$ 
by definition of the order function $w$.
In the following, each vertex $v \in w(f)$ will be assigned 
one of the following $d(v)+2$ colors
\[
 \{\top,0,1,2,\ldots,d(v)\}.
\]
The meaning of the color of a vertex $v$ is as follows:
for a  vertex set (possibly, a vector dominating set) $D$, 
\begin{itemize}
 \item $\top$ means that  $v \in D$.
 \item $i \in \{0,1,\ldots,d(v)\}$ means that  $v \notin D$
        and $|N_{G_f}(v)\cap D| \geq d(v)-i$.
\end{itemize}
Notice that a vertex colored by $i >0$ may need to be dominated
by some vertices in $V-V(G_f)$ for the feasibility.  
Given a coloring
$c \in \{\top,0,1,2,\ldots,d^*\}^{|w(f)|}$,
let $D_f(c) \subseteq V(G_f)$ be a vertex set with the minimum cardinality satisfying the following
(\ref{color1:eq})--(\ref{color3:eq}),
where $c(v)$ denotes the color assigned to a vertex $v \in V$:

\begin{eqnarray}
\label{color1:eq}
\!\!\!\!\!\!&& \mbox{$c(v)=\top$ if and only if $v \in D_f(c) \cap w(f)$.   }\\
\label{color2:eq}
\!\!\!\!\!\!&& \mbox{If $c(v)=i$, then
$v \in w(f)-D_f(c)$ and $|N_{G_f}(v)\cap D_f(c)|\geq d(v)-i$.} \\
\label{color3:eq}
\!\!\!\!\!\!&& 
\mbox{$|N_{G_f}(v)\cap D_f(c)|\geq d(v)$ holds for every vertex $v \in V(G_f)-w(f)-D_f(c)$.}
\end{eqnarray}

\noindent
Intuitively, $D_f(c)$ is a minimum vector dominating set in $G_f$ under the
assumption that the color for every vertex in $w(f)$ is restricted to
$c$.
Note that a vertex in $w(f)$ is  allowed  not to meet its demand in $G_f$,
because it can be dominated by some vertices in $V-V(G_f)$.
Also note that every  vertex in $V(G_f)-w(f)$ is not adjacent to 
any vertex in  $V-V(G_f)$ by Lemma~\ref{branch1:lem}(ii), and
it needs to be dominated by
 vertices only in $V(G_f)$ for the feasibility. 
We define  $A_f(c)$ as  $A_f(c)=|D_f(c)|$  if $D_f(c)$ exists and 
$A_f(c)=\infty$ otherwise.

Our dynamic programming algorithm proceeds bottom-up in $T$, while
computing
 $A_f(c)$ for all $c \in \{\top,0,1,2,\ldots,d^*\}^{|w(f)|}$
for each edge $f$ in $T$. 
We remark that
 $A_{(r_1,r_2)}(c)$ is the cardinality of a minimum vector dominating set, 
because $w(r_1,r_2)=\emptyset$ and $G_{(r_1,r_2)}=G$.
The algorithm consists of two types of procedures:
one is for leaf edges and the other is for non-leaf edges, 
where a {\em leaf edge} denotes an edge incident to a leaf of $T$.

\vspace*{0.2cm}

\noindent
{\bf Procedure for leaf edges:}
In the first step of the algorithm, we compute $A_f(c)$ for each edge
$f$ incident to a leaf of $T$.
Then, for all colorings
$c \in \{\top,0,1,2,\ldots,d^*\}^{|w(f)|}$,
let $A_f(c)$ be the number of vertices colored by $\top$ if $G_f$ and $c$
satisfy
(\ref{color1:eq}) -- (\ref{color3:eq}), and $A_f(c)=\infty$ otherwise.


For a fixed $c$, we need to  check if (\ref{color1:eq}) --
(\ref{color3:eq})
hold. This can be done in $O(|w(f)|)$ time.
Hence, this step takes $O((d^*+2)^{|w(f)|}|w(f)|)$ time.

\vspace*{0.2cm}

\noindent
{\bf Procedure for non-leaf edges:}
After the above initialization step, we visit  non-leaf edges of $T$ from
leaves to the root of $T$.
Let $f=(y_1,y_2)$ be a non-leaf edge of $T$ such that $y_1$ is the
parent of $y_2$,
 $y_3$ and $y_4$ are the children of $y_2$, and
$f_1=(y_2,y_3)$ and $f_2=(y_2,y_4)$.
Now we have already obtained $A_{f_j}(c')$ for all $c' \in 
 \{\top,0,1,2,\ldots,d^*\}^{|w(f_j)|}$, $j=1,2$.
By Lemma~\ref{branch1:lem}(i), we have $w(f) \subseteq w(f_1)\cup w(f_2)$,
 $w(f_1) \subseteq w(f_2)\cup w(f)$, and
 $w(f_2) \subseteq w(f)\cup w(f_1)$; 
let $X_1=w(f)- w(f_2)$,
$X_2=w(f)- w(f_1)$,
$X_3=w(f)\cap w(f_1) \cap w(f_2)$,
and
$X_4=w(f_1)-w(f)~(=w(f_2)-w(f))$.

We say that a coloring $c \in  \{\top,0,1,2,\ldots,d^*\}^{|w(f)|}$
of $w(f)$ is {\em formed} from a coloring $c_1$ of $w(f_1)$
and a coloring $c_2$ of $w(f_2)$ if the following (P1)--(P5) hold.

\vspace*{0.2cm}

\noindent
(P1) For every $v \in X_1 \cup X_2 \cup X_3$ with $c(v)=\top$,

\begin{quote}
 (a) If $v \in X_1 \cup X_3$, then $c_1(v)=\top$.\\
 (b) If $v \in X_2 \cup X_3$, then $c_2(v)=\top$.
\end{quote}

\noindent
(P2) For every $v \in X_4$,  $c_1(v)=\top$ if and only if $c_2(v)=\top$.

\vspace*{0.1cm}

\noindent
(P3) For every $v \in X_j-D_{c_1,c_2}$ 
where $\{j,j'\}=\{1,2\}$ and
 $D_{c_1,c_2}=\{v \in X_1 \cup X_2 \cup X_3 \cup X_4 \mid c_1(v)=\top \mbox{
or }c_2(v)=\top\}$, 
\begin{quote}
 If $c(v)=i$, then $c_j(v)=\min\{d(v),i+|D_{c_1,c_2}  \cap
 N_{G_f}(v)\cap X_{j'}|\}$.\\
(Intuitively, if $v\in X_j-D_{c_1,c_2}$ needs to be dominated by at least $d(v)-i$ vertices in
 $G_f$,
then  at least $\max\{0,d(v)-i-|D_{c_1,c_2}  \cap
 N_{G_f}(v)\cap X_{j'}|\}$ vertices from
 $V(G_{f_j})$ are necessary.)
\end{quote}

\noindent
(P4) For every $v \in X_3-D_{c_1,c_2}$, 
\begin{quote}
 If $c(v)=i$, then $c_1(v)=\min\{d(v),i+|D_{c_1,c_2}  \cap N_{G_f}(v)\cap X_2|+i_1\}$ and
 $c_2(v)=\min\{d(v),i+|D_{c_1,c_2} \cap N_{G_f}(v) \cap X_1 |+i_2\}$
for some nonnegative integers $i_1,i_2$
 with $i_1+i_2=\max\{0,d(v)-i-|D_{c_1,c_2}  \cap N_{G_f}(v) 
|\}$.\\
(Intuitively, if $v\in X_3-D_{c_1,c_2}$ needs to be dominated by at least $d(v)-i$ vertices in
 $G_f$,
then  at least $\max\{0,d(v)-i-|D_{c_1,c_2}  \cap N_{G_f}(v) 
|\}$ vertices  from 
$(V(G_{f_1})-w(f_1))\cup (V(G_{f_2})-w(f_2))$
are necessary for dominating $v$.
If  $i_1$ (resp., $i_2$) vertices among those vertices 
belong to $V(G_{f_2})-w(f_2)$ (resp., $V(G_{f_1})-w(f_1)$),
then 
at least $\max\{0,d(v)-i-|D_{c_1,c_2}  \cap
 N_{G_f}(v)\cap X_{j'}|-i_j\}$ vertices from $V(G_{f_j})$
are necessary for $\{j,j'\}=\{1,2\}$.)
\end{quote}

\noindent
(P5)  For every $v \in X_4-D_{c_1,c_2}$, 
\begin{quote}
  $c_1(v)=\min\{d(v),|D_{c_1,c_2}  \cap N_{G_f}(v)\cap X_2|+i_1\}$ and 
$c_2(v)= \min\{d(v),
|D_{c_1,c_2}  \cap N_{G_f}(v)\cap X_1|+i_2\}$
for some nonnegative integers $i_1,i_2$
 with $i_1+i_2=\max\{0,d(v)-|D_{c_1,c_2}  \cap N_{G_f}(v) 
|\}$. \\
(This case can be treated in a similar way to (P4).)
\end{quote}

As we will show in Lemmas~\ref{color1:lem} and \ref{color2:lem},
there exist 
 a coloring $c_1$ of $w(f_1)$
and a coloring $c_2$ of $w(f_2)$ forming $c$ such that
$D_{f_1}(c_1) \cup D_{f_2}(c_2)$ satisfies
(\ref{color1:eq})--(\ref{color3:eq})
and 
$|D_{f_1}(c_1) \cup D_{f_2}(c_2)| = A_f(c)$. 
Namely, we have
\[
 A_f(c)=\min\{ |A_{f_1}(c_1)|+|A_{f_2}(c_2)|- |D_{c_1,c_2}\cap (X_3\cup X_4)|
  \mid\mbox{ $c_1,c_2$ forms $c$}  \}.
\]
Thus, for all colorings $c \in \{\top,0,1,2,\ldots,d^*\}^{|w(f)|}$,
we can compute $A_f(c)$ from the information of $f_1$ and $f_2$.
By repeating these procedure bottom-up in $T$, we can
find a minimum vector dominating set in $G$.

Here, for a fixed $c$,
 we analyze the time complexity for computing $A_f(c)$.
Let $D_c=\{v \in w(f) \mid c(v)=\top\}$, 
$x_j=|X_j|$ for $j=1,2,3,4$, 
$z_3=|X_3-D_c|$, and
$z_4$ be the number of vertices in $X_4$ not colored by $\top$.
The number of pairs of a coloring $c_1$ of $w(f_1)$
and a coloring $c_2$ of $w(f_2)$ forming $c$ is at most 
\[
 (d^*+1)^{z_3} \sum_{z_4=0}^{x_4} {x_4 \choose z_4}  (d^*+1)^{z_4}  (d^*+1)^{z_4}
\]
since the number of pairs $(i_1,i_2)$ in (P4) or (P5) is at most
$d^*+1$ for each vertex in $X_3-D_c$ or each vertex in $X_4$ not colored
by $\top$.

Hence, for an edge $f$, the number of pairs forming $c$ is at most
\[
\begin{array}{l}
(d^*+2)^{x_1+x_2}\sum_{z_3=0}^{x_3} {x_3 \choose z_3} (d^*+1)^{z_3} 
 (d^*+1)^{z_3} \sum_{z_4=0}^{x_4} {x_4 \choose z_4}  (d^*+1)^{z_4}
 (d^*+1)^{z_4} \\[.2cm]
=(d^*+2)^{x_1+x_2} \{(d^*+1)^2+1\}^{x_3+x_4} 
\end{array}
\]
in total.
Now we have $x_1+x_2+x_3 \leq b$,
$x_1+x_3+x_4 \leq b$, and
$x_2+x_3+x_4 \leq b$ (recall that $b$ is the width of $(T',\tau)$).
By considering a linear programming problem which maximizes
$(x_1+x_2)\log(d^*+2)+(x_3+x_4)\log\{(d^*+1)^2+1\}$ subject to
these inequalities, 
we can observe that
$(d^*+2)^{x_1+x_2} \{(d^*+1)^2+1\}^{x_3+x_4} $
 attains the maximum when $x_1=x_2=x_4=b/2$ and $x_3=0$.
Thus, it takes in $O((d^*+2)^{b}\{(d^*+1)^2+1\}^{b/2})$ time to compute
$A_f(c)$ for all colorings $c$ of $w(f)$.

Since $|E(T)|=O(m)$ and the initialization step takes
$O((d^*+2)^{b}m)$
time in total, we can obtain $A_{(r_1,r_2)}(c)$ in $O(((d^*+2)^{b}\{(d^*+1)^2+1\}^{b/2}m)$time.

\begin{lemma}\label{color1:lem}
Let  $c \in  \{\top,0,1,2,\ldots,d^*\}^{|w(f)|}$ be a coloring of $w(f)$.
If a coloring $c_1$ of $w(f_1)$
and a coloring $c_2$ of $w(f_2)$ forms $c$, then
 $D_{f_1}(c_1) \cup D_{f_2}(c_2)$ satisfies
$(\ref{color1:eq})$--$(\ref{color3:eq})$
 for $f$.
\end{lemma}
\begin{proof}
We denote $D_{f_1}(c_1) \cup D_{f_2}(c_2)$ by $D'$,
and $D' \cap (X_1\cup X_2\cup X_3 \cup X_4)$ by 
 $D'_{c_1,c_2}$.
Clearly,  (\ref{color1:eq}) holds, since $v \in D' \cap
w(f)$
if and only if $c(v)=\top$ by the above (P1). 

We next show that $D'$ satisfies  (\ref{color2:eq}).
Let $v$ be a vertex in $X_1-D'=X_1-D'_{c_1,c_2}$.
From the above  (P3), we have 
$|N_{G_{f_1}}(v)\cap
 D'|\geq d(v)-i-|D'_{c_1,c_2}  \cap N_{G_f}(v)\cap X_2|$.
It follows that $|N_{G_f}(v)\cap D'|
\geq |N_{G_{f_1}}(v)\cap D'|+|D'_{c_1,c_2}  \cap N_{G_f}(v)\cap X_2|$
$\geq d(v)-i$.
Also,  the case of $v \in X_2-D'$ can be
 treated
similarly.

Let $v$ be a vertex in $X_3-D'=X_3-D'_{c_1,c_2}$.
Since $|N_{G_f}(v)\cap D'|$
$\geq |N_{G_f}(v)\cap D'_{c_1,c_2}|$ clearly holds, then we have only to consider
the case of $|N_{G_f}(v)\cap D'_{c_1,c_2}| < d(v)-i$.
From  (P4), we have 
$|N_{G_{f_1}}(v)\cap
 D'|\geq \max\{0,d(v)-i-|D'_{c_1,c_2}  \cap N_{G_f}(v)\cap X_2|-i_1\}$
and
$|N_{G_{f_2}}(v)\cap
 D'|\geq \max\{0,d(v)-i-|D'_{c_1,c_2}  \cap N_{G_f}(v)\cap X_1|-i_2\}$
where $i_1+i_2=d(v)-i-|D'_{c_1,c_2}  \cap N_{G_f}(v) 
|$ (note that $i_1+i_2>0$ from the  assumption of
 this case).
Notice that $(V(G_{f_1})-w(f_1))\cap (V(G_{f_2})-w(f_2))=\emptyset$ by
Lemma~\ref{branch1:lem}(ii).
It follows that $|N_{G_f}(v)\cap D'|\geq $
$|N_{G_{f_1}}(v)\cap D'|+|N_{G_{f_2}}(v)\cap D'|$
$-|N_{G_f}(v)\cap D'_{c_1,c_2} \cap (X_3 \cup X_4)|$
$\geq 2(d(v)-i)-|N_{G_f}(v)\cap D'_{c_1,c_2}|-i_1-i_2$
$=d(v)-i$.

We finally show that $D'$ satisfies (\ref{color3:eq}).
Let $v$ be a vertex in $X_4-D'$.
Since $|N_{G_f}(v)\cap D'|$
$\geq |N_{G_f}(v)\cap D'_{c_1,c_2}|$ clearly holds, 
then we have only to consider
the case of $|N_{G_f}(v)\cap D'_{c_1,c_2}| < d(v)$.
From  (P5), we have 
$|N_{G_{f_1}}(v)\cap
 D'|\geq \max\{0,d(v)-|D'_{c_1,c_2}  \cap N_{G_f}(v)\cap X_2|-i_1\}$
and
$|N_{G_{f_2}}(v)\cap
 D'|\geq \max\{0, d(v)-|D'_{c_1,c_2}  \cap N_{G_f}(v)\cap X_1|-i_2\}$
where $i_1+i_2=d(v)-|D'_{c_1,c_2}  \cap N_{G_f}(v)| >0$.
Hence, we have $|N_{G_f}(v)\cap D'|\geq $
$|N_{G_{f_1}}(v)\cap D'|+|N_{G_{f_2}}(v)\cap D'|$
$-|N_{G_f}(v)\cap D'_{c_1,c_2} \cap (X_3 \cup X_4)|$
$=2d(v)-|N_{G_f}(v)\cap D'_{c_1,c_2}|-i_1-i_2$
$=d(v)$.
Also, it follows from the definition of $D_{f_j}(c_j)$
 that $v \in V(G_{f_j})-w(f_j)$ satisfies  (\ref{color3:eq}) for $j=1,2$.
\qed
\end{proof}

\begin{lemma}\label{color2:lem}
Let  $c \in  \{\top,0,1,2,\ldots,d^*\}^{|w(f)|}$ be a coloring of $w(f)$.
There exist 
 a coloring $c_1$ of $w(f_1)$
and a coloring $c_2$ of $w(f_2)$ forming $c$ such that
$|D_{f_1}(c_1) \cup D_{f_2}(c_2)| \leq A_f(c)$.  
\end{lemma}
\begin{proof}
For each vertex $v \in w(f_j)$, $j=1,2,$ let\\ 
\[
 c_j(v)= \left\{   
 \begin{array}{ll}
  \top & \mbox{ if }v \in D_f(c),\\
 \min\{d(v),c(v)+|N_{G_f}(v)\cap D_f(c) -V(G_{f_j})|\} & \mbox{ if }v \in X_j-D_f(c), \\
\max\{0,d(v)-|N_{G_{f_j}}(v)\cap
 D_f(c)|\} & \mbox{ if }v \in X_3 \cup X_4-D_f(c).
 \end{array}
\right.
\]
For $v \in X_j-D_f(c)$, 
we have
$|N_{G_f}(v) \cap D_f(c)|=$
 $|N_{G_{f_j}}(v) \cap D_f(c)| + |N_{G_f}(v)\cap D_f(c) -V(G_{f_j})| \geq d(v)-c(v)$,
since $D_f(c)$ satisfies (\ref{color2:eq}).
Hence, $|N_{G_{f_j}}(v) \cap D_f(c)|\geq 
\max\{0,d(v)-c(v)-|N_{G_f}(v)\cap D_f(c) -V(G_{f_j})|\}=
d(v)-c_j(v)$ for
all $v \in w(f_j)-D_f(c)$.
It follows from that 
 the minimality of $A_{f_j}(c_j)$ implies that
$|D_f(c) \cap V(G_{f_j})|\geq A_{f_j}(c_j)$; hence, 
$A_f(c) \geq |D_{f_1}(c_1) \cup D_{f_2}(c_2)|$.  
On the other hand, $c_1$ and $c_2$ does not necessarily form $c$.
Below, we show that there exist a coloring $c_1'$ of $w(f_1)$ 
and a coloring  $c_2'$ of $w(f_2)$ forming $c$ such that
$c_j'(v)\geq c_j(v)$ for every $v \in w(f_j)-D_{f}(c)$ for $j=1,2$.
%
Note that  $D_{f_j}(c_j)$ satisfies (\ref{color1:eq})--(\ref{color3:eq}) 
also for
 $c_j'$,
since $|N_{G_{f_j}}(v) \cap D_{f_j}(c)|\geq d(v)-c_j(v)\geq d(v)-c_j'(v)$
for every $v \in w(f_j)-D_{f}(c)$. 
Hence, from the minimality of $|D_{f_j}(c_j')|$, 
we have 
$A_f(c) \geq |D_{f_1}(c_1) \cup D_{f_2}(c_2)| \geq |D_{f_1}(c_1') \cup
 D_{f_2}(c_2')|$, which proves this lemma.

We can construct such $c_1'$, $c_2'$ as follows.
First let $c_j'(v)=c_j(v)$ for all $v \in X_1 \cup X_2 \cup D_f(c)$;
$c'_1$ and $c_2'$ satisfy (P1) and (P2)
 in the definition of a coloring $c$ formed by $c_1$ and $c_2$.
By Lemma~\ref{branch1:lem}(ii), every $v \in X_j$ satisfies 
$N_{G_f}(v)\cap D_f(c) -V(G_{f_j})=N_{G_f}(v)\cap D_f(c) \cap X_{j'}$
for $\{j,j'\}=\{1,2\}$.
Hence, $c_j'(v) (=c_j(v))$ for $v \in X_j-D_f(c)$, $j=1,2$ satisfies 
 (P3).

Let $v \in X_3-D_f(c)$. 
Since $D_f(c)$ satisfies (\ref{color2:eq}), 
we have $|N_{G_f}(v)\cap D_f(c)|\geq d(v)-c(v)$.
Now from construction of $c_1$ and $c_2$,
the value $i_1'$ (resp., $i_2'$) corresponding to $i_1$ (resp., $i_2$) in (P4)
in the definition of  $c$ formed by $c_1$ and $c_2$
is
$\max\{0,d(v)-|N_{G_{f_1}}(v)\cap D_f(c)|-c(v)-|N_{G_f}(v)\cap X_2 \cap D_f(c)|\}$
(resp., 
$\max\{0,d(v)-|N_{G_{f_2}}(v)\cap D_f(c)|-c(v)-|N_{G_f}(v)\cap X_1 \cap D_f(c)|\}$).
It follows that $i_1'+i_2'  $
$\leq \max\{0,d(v)-c(v)-|N_{G_f}(v)\cap D_f(c) \cap (X_1 \cup X_2 \cup X_3 \cup
 X_4)|\}$ (note that the final inequality follows from $|N_{G_f}(v)\cap D_f(c)|\geq d(v)-c(v)$).

Let $v \in X_4-D_f(c)$. 
Since $D_f(c)$ satisfies (\ref{color2:eq}), 
we have $|N_{G_f}(v)\cap D_f(c)|\geq d(v)$.
From construction of $c_1$ and $c_2$,
the value $i_1'$ (resp., $i_2'$) corresponding to $i_1$ (resp., $i_2$) in (P5)
in the definition of  $c$ formed by $c_1$ and $c_2$
is
$\max\{0,d(v)-|N_{G_{f_1}}(v)\cap D_f(c)|-|N_{G_f}(v)\cap X_2 \cap D_f(c)|\}$
(resp., 
$\max\{d(v)-|N_{G_{f_2}}(v)\cap D_f(c)|-|N_{G_f}(v)\cap X_1 \cap D_f(c)|\}$). 
It follows that $i_1'+i_2'  \leq \max\{0,d(v)-|N_{G_f}(v)\cap D_f(c) \cap (X_1 \cup X_2 \cup X_3 \cup
 X_4)|\}$.

Consequently, we can construct
a coloring $c_1'$ of $w(f_1)$ 
and a coloring  $c_2'$ of $w(f_2)$ forming $c$ such that
$c_j'(v)\geq c_j(v)$ for every $v \in X_3 \cup X_4-D_f(c)$ 
and
$c_j'(v)= c_j(v)$ for every $v \in D_f(c) \cup X_1 \cup X_2$ 
for $j=1,2$
by increasing $i_1'$ or $i_2'$ for each vertex  $v \in X_3 \cup X_4-D_f(c)$
so that $i_1'+i_2'$ becomes equal to
$\max\{0,d(v)-c(v)-|N_{G_f}(v)\cap D_f(c) \cap (X_1 \cup X_2 \cup X_3 \cup
 X_4)|\}$ 
(resp., $\max\{0,d(v)-|N_{G_f}(v)\cap D_f(c) \cap (X_1 \cup X_2 \cup X_3 \cup
 X_4)|\}$)
if  $v \in X_3$ (resp., $v \in X_4$).
\qed
\end{proof}

Summarizing the  arguments given so far, we have shown Theorem~\ref{algo1:th}.

\subsection{Total vector domination and multiple domination}\label{algo2:subsec}

We consider the total vector domination problem.
The difference between the total vector domination and the vector  domination
is that each vertex selected as a member in  a dominating set needs to be
dominated or not.
Hence, we will
 modify the following parts (I)--(III) in the algorithm
for vector domination given in the previous subsection 
so that
each vertex selected as a member in a dominating set also satisfies its
demand.

\vspace*{0.2cm}

(I) Color assignments:
Let $f \in E(T)$ be an edge in a branch decomposition $T$ of $G$.
We will assign  to each vertex $v \in w(f)$  an ordered pair $(\ell, i)$ 
 of colors, $\ell \in \{\top, \bot \}$, $i \in \{0,1,\ldots,d(v)\}$,
 where  $\top$  means that $v$ belongs
  to the dominating set,
$\bot$  means that $v$ does not belong
  to the dominating set, and
and $i$ means that $v$ is dominated by at least $d(v)-i$ vertices
in $G_f$.

\vspace*{0.1cm}

(II) Conditions for $D_f(c)$:
For a coloring $c \in (\{\top,\bot\}\times\{0,1,2,\ldots,$
$d^*\})^{|w(f)|}$,
we modify (\ref{color1:eq})--(\ref{color3:eq}) as follows,
where let $c(v)=(c^1(v),c^2(v))$:

\begin{eqnarray}
\label{color2-1:eq}
\nonumber
\!\!\!\!\!\!&& \mbox{$c^1(v)=\top$ if and only if $v \in D_f(c) \cap w(f)$.  }\\
\label{color2-2:eq}
\nonumber
\!\!\!\!\!\!&& \mbox{If $c^2(v)=i$, then  $|N_{G_f}(v)\cap D_f(c)|\geq d(v)-i$.} \\
\label{color2-3:eq}
\nonumber
\!\!\!\!\!\!&& 
\mbox{$|N_{G_f}(v)\cap D_f(c)|\geq d(v)$ holds for every vertex $v \in V(G_f)-w(f)$.}
\end{eqnarray}

(III) Definition of a coloring $c$ formed by $c_1$ and $c_2$:
For a coloring $c \in (\{\top,\bot\}\times\{0,1,2,\ldots,$
$d^*\})^{|w(f)|}$,
we modify (P1)--(P5) as follows:

\vspace*{0.2cm}

\noindent
(P1') For every $v \in X_1 \cup X_2 \cup X_3$ with $c^1(v)=\top$ (resp.,
$c^1(v)=\bot$),

\begin{quote}
 (a) If $v \in X_1 \cup X_3$, then $c^1_1(v)=\top$ (resp., $c^1_1(v)=\bot$).\\
 (b) If $v \in X_2 \cup X_3$, then $c^1_2(v)=\top$ (resp., $c^1_2(v)=\bot$).
\end{quote}

\noindent
(P2') For every $v \in X_4$,  $c^1_1(v)=\top$ (resp., $c^1_1(v)=\bot$)
if and only if $c^1_2(v)=\top$ (resp., $c^1_2(v)=\bot$).

\vspace*{0.1cm}

\noindent
(P3') For every $v \in X_j$ 
where $\{j,j'\}=\{1,2\}$ and
 $D_{c_1,c_2}=\{v \in X_1 \cup X_2 \cup X_3 \cup X_4 \mid c^1_1(v)=\top \mbox{
or }c^1_2(v)=\top\}$, 
\begin{quote}
 If $c^2(v)=i$, then $c^2_j(v)=\min\{d(v),i+|D_{c_1,c_2}  \cap
 N_{G_f}(v)\cap X_{j'}|\}$.
\end{quote}

\noindent
(P4') For every $v \in X_3$, 
\begin{quote}
 If $c^2(v)=i$, then $c^2_1(v)=\min\{d(v),i+|D_{c_1,c_2}  \cap N_{G_f}(v)\cap X_2|+i_1\}$ and
 $c^2_2(v)=\min\{d(v),i+|D_{c_1,c_2} \cap N_{G_f}(v) \cap X_1 |+i_2\}$
for some nonnegative integers $i_1,i_2$
 with $i_1+i_2=\max\{0,d(v)-i-|D_{c_1,c_2}  \cap N_{G_f}(v) \cap (X_1
 \cup X_2 \cup X_3 \cup X_4)|\}$.
\end{quote}

\noindent
(P5')  For every $v \in X_4$, 
\begin{quote}
  $c^2_1(v)=\min\{d(v),|D_{c_1,c_2}  \cap N_{G_f}(v)\cap X_2|+i_1\}$ and 
$c^2_2(v)= \min\{d(v),
|D_{c_1,c_2}  \cap N_{G_f}(v)\cap X_1|+i_2\}$
for some nonnegative integers $i_1,i_2$
 with $i_1+i_2=\max\{0,d(v)-|D_{c_1,c_2}  \cap N_{G_f}(v) \cap (X_1 \cup
 X_2 \cup X_3 \cup X_4)|\}$. 
\end{quote}


We analyze the time complexity of this modified algorithm.
Similarly to the case of the vector domination, 
the total running time is dominated by total complexity  for computing 
$A_f(c)$ for non-leaf edges $f$.

Let $f$ be a non-leaf edge of $T$ and 
$x_i$, $i=1,2,3,4$ and $z_4$ be defined as the previous subsection.
The number of pairs of a coloring $c_1$ of $w(f_1)$
and a coloring $c_2$ of $w(f_2)$ forming $c$ is at most 
\[
 (d^*+1)^{x_3} \sum_{z_4=0}^{x_4} {x_4 \choose z_4}  (d^*+1)^{x_4}  (d^*+1)^{x_4}
\]
since the number of pairs $(i_1,i_2)$ in (P4') or (P5') is at most
$d^*+1$ for each vertex in $X_3 \cup X_4$.
Hence, for an edge $f$, the number of pairs forming $c$ is at most
\[
\begin{array}{l}
\{2(d^*+1)\}^{x_1+x_2}\sum_{z_3=0}^{x_3} {x_3 \choose z_3} (d^*+1)^{x_3} 
 (d^*+1)^{x_3} \sum_{z_4=0}^{x_4} {x_4 \choose z_4}  (d^*+1)^{x_4}
 (d^*+1)^{x_4} \\[.2cm]
=\{2(d^*+1)\}^{x_1+x_2} \{2(d^*+1)^2\}^{x_3+x_4} 
\end{array}
\]
in total.
Since $x_1+x_2+x_3 \leq b$,
$x_1+x_3+x_4 \leq b$, and
$x_2+x_3+x_4 \leq b$,
it follows that
$(x_1+x_2)\log(2d^*+2)+(x_3+x_4)\log\{2(d^*+1)^2\}$ 
 attains the maximum 
when $x_1=x_2=x_4=b/2$ and $x_3=0$.
Thus, it takes in $O(2^{3b/2}(d^*+1)^{2b})$ 
 time to compute
$A_f(c)$ for all colorings $c$ of $w(f)$.
Namely, we obtain the following theorem.

\begin{theorem}\label{algo2:th}
If  a branch decomposition of $G$ with width $b$ is given, 
a minimum total vector dominating set in $G$  can be found in 
 $O(2^{3b/2}(d^*+1)^{2b}m)$ 
time. 
\qed
\end{theorem}

Also, by replacing $N_G()$ with $N_G[]$ in the modification for
total vector domination, we can
obtain the following theorem for the multiple domination problems.

\begin{theorem}\label{algo3:th}
If  a branch decomposition of $G$ with width $b$ is given, 
a minimum multiple dominating set in $G$  can be found in 
 $O(2^{3b/2}(d^*+1)^{2b}m)$ 
time. 
\qed
\end{theorem}

\section{Subexponential fixed parameter algorithm for planar graphs}\label{fpt:sec}

We consider the problem of 
checking whether a given graph $G$ 
has a $d$-vector dominating set with cardinality
at most $k$.
As mentioned in Subsection~\ref{related:subsec},
if $G$ is $\rho$-degenerated, then the problem can be solved in
$k^{O(\rho k^2)}n^{O(1)}$ time. 
Since a planar graph is 5-degenerated, 
 it follows that 
the problem with a planar graph can be solved in
$k^{O(k^2)}n^{O(1)}$ time.
In this section, we give a subexponential fixed-parameter algorithm, 
parameterized by $k$,  
 for a planar graph; namely,
we will show the following theorem.

\begin{theorem}\label{fpt1:th}
If $G$ is a planar graph, then we can check in
 $O(n^4+(\min\{d^*,k\}+2)^{b^*}\{(\min\{d^*,k\}+1)^2+1\}^{b^*/2}n)$
 time whether $G$ has a $d$-vector dominating set with cardinality at
 most $k$ or not, where $b^*=\min\{12\sqrt{k+z}+9,20\sqrt{k}+17\}$ and
 $z=|\{v\in V \mid d(v)=0\}|$. 
 \qed
\end{theorem}

\noindent
This time complexity is roughly $O(n^4+2^{O(\sqrt{k}\log k)}n)$, which 
is subexponential with respect to $k$; this improves the running time of
the previous fixed-parameter algorithm.  

Let $V_0=\{v \in V \mid d(v)=0\}$ and $z=|V_0|$. 
In \cite[Lemma 2.2]{fomin2006dominating},
it was shown that if a planar graph $G'$ has an ordinary dominating set 
(i.e., a (1,1,\ldots,1)-vector dominating set) with cardinality
at most $k$, then $bw(G') \leq 12\sqrt{k}+9$. This bounds is based on
the {\em bidimensionality}~\cite{demaine2008bidimensionality}, 
and was used to design the subexponential fixed-parameter algorithm 
with respect to $k$ for the ordinary dominating set problem. In the case
of our domination problems, however, it is difficult to say that they
have the bidimensionality, due to the existence of $V_0$ vertices. Instead, 
we give a similar bound on the branchwidth not w.r.t $k$ but w.r.t $k+z$
as follows: For any (total, multiple) $d$-vector dominating set $D$ of
$G$ ($|D|\le k$), $D \cup V_0$ is an ordinary dominating set of $G$, and this yields
$bw(G) \leq 12\sqrt{k+z}+9$.  

Actually, it is also possible to exclude $z$ from the parameters, though the
coefficient of the exponent becomes larger. To this end, we use the
notion of $(k,2)$-center. Recall that a $(k,r)$-{\em center} of $G'$ is
a set $W$ of vertices of $G'$ with size $k$ such that any vertex in $G'$
is within distance $r$ from a vertex of $W$. 
For a $(k,r)$-center, a similar bound on
the branchwidth is known: if a planar graph $G'$ has a $(k,r)$-center, 
then $bw(G') \leq 4(2r+1)\sqrt{k}+8r+1$ (\cite[Theorem
3.2]{demaine2005fixed}). Here, we use this bound. We can assume that for
$v\in V_0$, $N_G(v) \not\subseteq V_0$ holds, because 
$v\in V_0$
satisfying $N_G(v) \subseteq V_0$ is never selected as a member of any
optimal solution; it is {\em irrelevant}, 
and we can remove it. 
That is, every vertex in $V_0$ has at least one neighbor from
$V-V_0$. Then, for any (total, multiple) $d$-vector dominating set $D$
of $G$ ($|D|\le k$), $D$ is a $(k,2)$-center of $G$. This is because any
vertex in $V-V_0$  is adjacent to a vertex in $D$ and any vertex in
$V_0$ is adjacent to a  vertex in $V-V_0$. Thus, we have $bw(G) \leq
20\sqrt{k}+17$.   

In summary, we have the following lemma.

\begin{lemma}\label{lem:bidim}
Assume that  $G$ is a planar graph without irrelevant vertices,
i.e., $N_G(v) \not\subseteq V_0$ holds for each $v \in V_0$.
 Then, if
 $G$ has a (total, multiple) vector dominating set with
 cardinality at most $k$,
then  we have $bw(G) \leq
\min\{12\sqrt{k+z}+9,20\sqrt{k}+17\}$. \qed 
\end{lemma}

\noindent
Combining this lemma with  the algorithm in
Subsection~\ref{algo1:subsec}, we can check whether a given graph
has  a  vector
dominating
set with cardinality at most $k$ according to the following steps 1 and 2:

\vspace*{0.2cm}

\noindent
Step 1: Let $b^*=\min\{12\sqrt{k+z}+9,20\sqrt{k}+17\}$. Check whether
the branchwidth of $G$ is at most $b^*$. 
If so, then  go to Step 2, and otherwise halt after outputting ^^ NO'.

\vspace*{0.2cm}

\noindent
Step 2: Construct a branch decomposition with width at most $b^*$,
and apply the dynamic programming algorithm in
Subsection~\ref{algo1:subsec}
to find a minimum vector dominating set for $G$.

\vspace*{0.2cm}

By  Lemma~\ref{branch:lem},  Theorem~\ref{algo1:th}, and
the fact that  
 any planar graph $G'$ satisfies  $|E(G')|=O(|V(G')|)$,
it follows that the running time of this procedure is
 $O(n^4+(d^*+2)^{b^*}\{(d^*+1)^2+1\}^{b^*/2}n)$.
Hence, in the case of $d^* \leq k$,  Theorem~\ref{fpt1:th} has been proved.

The  case of $d^* > k$ can be reduced to the case of $d^* \leq k$ by
 the following standard  kernelization method, 
which proves Theorem~\ref{fpt1:th}.
Assume that  $d^* > k$.
Let $V_{\max}(d)$ be the set of vertices $v$ with $d(v)=d^*$.
For the feasibility, we need to select each  vertex
$v \in V_{\max}(d)$  as a member in a vector dominating set.
Hence, if  $|V_{\max}(d)| > k$,
 then it turns out that $G$ has no 
 vector dominating set with cardinality
at most $k$.
Assume that  $|V_{\max}(d)| \leq  k$.
Then, it is not difficult to see that 
we can reduce an instance
$I(G,d,k)$  with $G$,  $d$, and $k$ to an instance
$I(G',d',k')$ 
such that $G'=G-V_{\max}(d)$ (i.e., $G'$ is the graph obtained from $G$ by deleting $V_{\max}(d)$),
$d'(v)=\max\{0,d(v)-|N_G(v)\cap V_{\max}(d)|\}$ for all vertices $v \in
V(G')$,
and $k'=\max\{0,k-|V_{\max}(d)|\}$.
Based on this observation, we can reduce $I(G,d,k)$
to an instance $I(G'',d'',k'')$ with $\max\{d''(v)\mid v \in V(G'')\}
\leq
k'' \leq k$ 
or output ^^ YES' or ^^ NO'
in the  following manner:

\vspace*{0.2cm}

\noindent
(a) After setting $G':=G$, $d':=d$, and $k':=k$, repeat the procedures (b1)--(b3)
while $k' < {d'}^* (=\max\{d'(v)\mid v \in V(G')\})$.

\vspace*{0.2cm}

\noindent
(b1) If $k' <  |V_{\max}(d')|$, then halt after outputting ^^ NO.'

\noindent
(b2) If $k' \geq |V_{\max}(d')|$ and  $V(G')=V_{\max}(d')$, 
then halt after outputting ^^ YES.' 

\noindent
(b3) Otherwise after setting
$G'':=G'-V_{\max}(d')$,
$d''(v):=\max\{0,d'(v)-|N_{G'}(v)\cap V_{\max}(d')|\}$ for each $v \in
V(G'')$,
and 
$k'':=\max\{0,k'-|V_{\max}(d')|\}$, redefine $G''$, $d''$, and $k''$
as $G'$, $d'$, and $k'$, respectively.

\vspace*{0.2cm}

Next, we consider the total vector domination problem
and the multiple domination problem.
For these problems, since all vertices $v \in V$ need to be dominated by
$d(v)$ vertices, the condition that $d^* \leq k$ is necessary for the feasibility.
Similarly, we have the following theorem by Theorems~\ref{algo2:th} and
\ref{algo3:th}. 

\begin{theorem}\label{fpt2:th}
Assume that a given graph $G$ is planar, and let
 $b^*=\min\{12\sqrt{k+z}+9,20\sqrt{k}+17\}$. \\
$(i)$ We can check in 
 $O(n^4+2^{3b^*/2}(\min\{d^*,k\}+2)^{2b^*}n)$
time 
whether $G$ has a total vector dominating set with cardinality at most $k$
or not. \\
$(ii)$
We can check in 
 $O(n^4+2^{3b^*/2}(\min\{d^*,k\}+2)^{2b^*}n)$
time 
whether $G$ has a multiple  dominating set with cardinality at most $k$
or not.
\qed
\end{theorem}

Before concluding this section, we mention that the above result 
can be extended to apex-minor-free graphs, a superclass of planar
graphs. For apex-minor-free graphs, the following lemma is known. 

\begin{lemma}(\cite[Lemma 2]{fomin2011subexponential})
Let $G$ be an apex-minor-free graph. If $G$ has a $(k,r)$-center,
then the treewidth of $G$ is $O(r\sqrt{k})$. 
\end{lemma}

\noindent 
\nop{
Also, from the constant-factor approximation algorithm
of
 \cite{demaine2005decomposition} for computing a tree-decomposition
of minimum width in 
$H$-minor-free graphs, which is a superclass of apex-minor-free graphs, and the property that 
a tree-decomposition width $t$ can be converted to
a branch-decomposition with at most $t+1$ in polynomial time \cite{robertson1991graph},
we have the following lemma.

\begin{lemma}
Let $G$ be an $H$-minor-free graph.
Then, a branch-decomposition of $G$
with width $O(bw(G))$ can be computed in polynomial time 
 for any fixed $H$.
\end{lemma}
}

From this lemma, the linear relation of treewidth and branchwidth, 
and the $2^{O(bw(G))}n^2$ -time algorithm  
for computing a branch decomposition with width $O(bw(G))$ (mentioned
after Theorem~\ref{algo1:th}), 
we obtain the following corollary. 

\begin{corollary}
Let $G$ be an apex-minor-free graph. We can check in $2^{O(\sqrt{k}\log
 k)} n^{O(1)}$ time whether $G$ has a (total, multiple) vector
 dominating set with cardinality at most $k$ or not. 
\end{corollary}

\nop{
\begin{remark}
It is natural to try to design a subexponential fixed-parameter algorithm 
for the domination problems of planar graphs with respect to only
 solution size $k$.  
One of general techniques to design a subexponential fixed-parameter algorithm 
for planar graphs or wider classes of graphs is to utilize the
bidimensionality~\cite{demaine2008bidimensionality}. 
The bidimensionality is the idea behind \cite[Lemma
 2.2]{fomin2006dominating}, which is used to prove Lemma 
 \ref{lem:bidim}. In fact, the $O(2^{15.13\sqrt{k}}+n^3)$-time algorithm 
 for planar graphs is designed based on \cite[Lemma
 2.2]{fomin2006dominating}, i.e., by the bidimensionality of the ordinary 
dominating set. 
In the cases of our domination problems, however, 
it seems to be difficult to apply the bidimensionality technique
with removing parameter $z$. For example, the transformation
 discussed in \cite{betzler2012bounded} can reduce the domination problems 
with $0$ demands (i.e., $z>0$) to positive cases (i.e., $z=0$), but the
 transformation does not preserve the bidimensionality. In the sense, we
might need a more elaborate technique to obtain  a subexponential
 fixed-parameter algorithm for the domination problems of planar graphs
 with respect to only solution size $k$.  
\end{remark}
}

\bibliographystyle{abbrv}
\bibliography{domination}
%
%
%
%
%
%

\end{document}